\documentclass[a4paper,aps,prd,preprint,showkeys,nofootinbib,superscriptaddress,11pt]{revtex4-1}
\usepackage[top=2.5cm, bottom=2.5cm, left=2.5cm, right=2.5cm]{geometry}
\usepackage{amsfonts}
\usepackage{amssymb}
\usepackage{amsmath}
\usepackage{amsthm}
\usepackage{graphicx}
\usepackage{indentfirst}
\usepackage{slashed}
\usepackage{tensor}
\usepackage{caption}
\usepackage[caption=false]{subfig}
\usepackage{dsfont}
\usepackage{enumitem}
\usepackage{hyperref}
\usepackage[utf8]{inputenc}
\usepackage{color}

\newcommand{\ii}{\ensuremath{\mathrm{i}}}

\newcommand{\e}{\ensuremath{\mathrm{e}}}

\newcommand{\R}{\ensuremath{\mathrm{R}}}

\newcommand{\Arctan}{\ensuremath{\mathrm{Arctan\!}}}
\definecolor{darkgreen}{rgb}{0,0.35,0}
\definecolor{Rood}{rgb}{1, 0, 0}

\newtheorem{theorem}{Theorem}[section]

\newtheorem{lemma}[theorem]{Lemma}

\newcommand{\cecs}{Centro de Estudios Cient\'{\i}ficos (CECs), Casilla 1469, Valdivia, Chile}
\newcommand{\ucharles}{Faculty of Mathematics and Physics, Charles University, V Hole\v{s}ovi\v{c}k\'ach 2, 18000 Prague 8, Czech Republic}
\newcommand{\kulak}{Department of Physics, KU Leuven Campus Kortrijk--Kulak, Etienne Sabbelaan 53 Bs 7657, 8500 Kortrijk, Belgium}
\newcommand{\ELI}{Institute of Physics of the ASCR, ELI Beamlines Project, Na Slovance 2, 18221 Prague, Czech Republic}
\newcommand{\infn}{ INFN, Sezione di Napoli, Complesso Universitario di Monte S.~Angelo, Via Cintia Edificio 6, 80126 Naples, Italia}
\newcommand{\napoli}{Dipartimento di Matematica e Applicazioni "R. Caccioppoli", Universit\'{a} di Napoli Federico II, Complesso Universitario di Monte S.~Angelo,  Via Cintia Edificio 6,
80126 Naples, Italia }
\newcommand{\uach}{Universidad Austral de Chile (UACh), Campus Isla Teja, Ed. Emilio Pugin}
\newcommand{\ughent}{Ghent University, Department of Physics and Astronomy, Krijgslaan 281-S9, 9000 Ghent, Belgium}
\newcommand{\uff}{Instituto de F\'isica, Universidade Federal Fluminense, Av. Litoranea s/n, 24210-346, Niter\'oi, RJ, Brasil}

\begin{document}

\title{Extracting topological information from momentum space propagators}

\author{Fabrizio Canfora}
\email{fcanforat@gmail.com}
\affiliation{\cecs}
\author{David Dudal}
\email{david.dudal@kuleuven.be}
\affiliation{\kulak}
\affiliation{\ughent}
\author{Alex Giacomini}
\email{alexgiacomini@uach.cl}
\affiliation{\uach}
\author{Igor F.~Justo}
\email{igorfjusto@gmail.com}
\affiliation{\uach}
\affiliation{\uff}
\author {Pablo Pais}
\email{pais@ipnp.troja.mff.cuni.cz}
\affiliation{\ucharles}
\affiliation{\ELI}
\author{Luigi Rosa}
\email{rosa@na.infn.it}
\affiliation{\napoli}
\affiliation{\infn}

\begin{abstract}
A new topological invariant quantity, sensitive to the analytic structure of both fermionic and bosonic propagators, is proposed. The gauge invariance of our construct is guaranteed for at least small gauge transformations. A generalization compatible with the presence of complex poles is introduced and applied to the classification of propagators typically emerging from non-perturbative considerations. We present partial evidence that the topological number can be used to detect chiral symmetry breaking or deconfinement.
\end{abstract}

\maketitle


\section{Introduction}
\label{Section_Introduction}

Topology  has been used since a long time in the study of condensed matter physics and, nowadays, it has become the main theoretical
tool for the description and classification of topological insulators (see \cite{Qi:2011zya} for a comprehensive review). This is because it is possible to construct topological invariant quantities, i.e.~, quantities that are invariant under ``smooth deformations'' (homeomorphisms). In this manner the space of parameter-dependent results is naturally divided into disjoint sectors, invariant with respect to homeomorphisms, that are characterized by different values for the topological invariants. Also in a particle physics, or more general, (quantum) field theory context, the importance of topology cannot be underestimated, see \cite{Rajaraman:1982is} for the typical illustrative examples.

From the physical point of view, different sectors correspond to different  physical phases. Similarly, momentum space topology (MST) has been applied to the classification of the ground state of relativistic quantum field theories (see \cite{So:1985wv}, where MST is first applied to lattice fermions) into universality classes, however without the same resonance as in condensed matter \cite{Horava:2005jt,Volovik:2011kg,Volovik:2009au,Volovik:2006gt,Zubkov:2012ws,Zubkov:2016llc}. The MST framework might be a very powerful tool and has also been applied to the investigation of emergent gravity, \cite{Klinkhamer:2007pe,Volovik:2009av}. Indeed relativistic quantum fields share some topological properties with topological materials \cite{Horava:2005jt}. The Standard Model (SM) below the Electroweak scale, for example,  belongs to the same universality class as of the superfluid $^{3}$He and of the three-dimensional topological insulators obeying time-reversal symmetry \cite{Volovik:2011kg}. However, up to now, topological invariants constructed in MST have not been applied to bosonic systems. In a sense, this could be interpreted as one of the major drawbacks of MST, which only applies to ``half of the world'', for reasons that will become clear in the next sections. On the contrary, it would be desirable  to apply MST to bosonic fields, especially considering that the analytic structure in momentum space of bosonic propagators might encode extremely valuable physical information.

Inspired by the power of MST in the classification of the vacuum of quantum field theories, we propose a new topologically invariant object that is not only sensitive to the full analytic structure of fermionic propagators, but that can also be applied to bosonic field propagators. Particularly, we will consider here the case of certain classes of fits to lattice fermion propagators \cite{Furui:2006ks,Dudal:2013vha,Burgio:2012ph,Rojas:2013tza} and gluon propagators \cite{Dudal:2010tf,Cucchieri:2011ig,Dudal:2012zx,Dudal:2018cli}, although the proposed tool can be applied to any relativistic fermionic or bosonic propagator. Hence, the present results give rise to a considerable enhancement of the power, as well as applicability, of MST. The paper is organized as follows: in Section \ref{Section_momentum_space_TI} we recall the definition and properties of the momentum space topological invariant. In Section \ref{Section_fermionic} we extract, departing from the standard definition, a specific representation for the MST invariant that is very useful for a generalization to the bosonic sector. In Section \ref{Section_Generalized_TI}, we generalize this definition to make it compatible with the potential presence of complex poles and we apply this latest definition to several concrete propagators in Section \ref{Section_Applications}. Finally, in Section \ref{Section_Conclusions}, our conclusions can be found.

\section{The momentum space topological invariant}
\label{Section_momentum_space_TI}

Let us start with the following topological quantity, to our knowledge first introduced by So \cite{So:1985wv}, see also \cite{Zubkov:2016llc,Katsura:2016wqa,Volovik:2016mre},
\begin{equation}
\mathcal{N}_{3}~=~N\mathrm{Tr}\int_{\Sigma }K~G\mathrm{d}G^{-1}\wedge G\mathrm{d}G^{-1}\wedge
G\mathrm{d}G^{-1}\;,
\label{topinv}
\end{equation}%
where $K$ stands for the matrix representation of a vacuum symmetry; $G$ is the two-point Green function of the fermion field; the integral is on a
three/dimensional hypersurface $\Sigma $, in our case defined by imposing $p_{0}=0$. From now on, we assume the following notation: $p^{2}=p_{i}p^{i}$ and
$\slashed{p}=\gamma_{i}p^{i}$, with $i=1,2$ and $3$ (the spatial indices); finally, $N$ is a normalization factor. Notice that the foregoing implies that $[K,G]=0$. A rather similar, albeit not exactly the same, topological invariant was also studied in the lattice context of \cite{Coste:1989wf}.

At first glance, the topological invariant defined above can be constructed only for fermion propagators, since the $\gamma $-matrices are indispensable to
get a non-trivial result. A direct computation reveals that if one tries to make sense of the above expression for bosonic propagators with internal indices, the result vanishes identically while for fermions the $\gamma$-matrices save the day by generating the completely antisymmetric tensor.

Our main goal now for the following sections is to propose a new, and more general, topological invariant
that can also be applied to bosonic fields.

If there would be no interactions between the fermions, the quantity \eqref{topinv} becomes nothing else than the number of massive
flavors of Dirac fermions \cite{Zubkov:2012ws}. Adding interactions in such a way that these do not spoil the condition $[G,K]=0$, the change of Eq.~\eqref{topinv} under a continuous deformation $G\to G + \delta G$, where
$\delta G$ encodes the variations of the parameters in $G$, can be written as
\begin{equation*}
\delta \mathcal{N}_{3} \;=\; 3\; N\; \int_{\Sigma} \mbox{Tr}\left[K \; \delta[G \mathrm{d}G^{-1}] \wedge G\mathrm{d}G^{-1} \wedge G\mathrm{d}G^{-1}\right] \;,
\end{equation*}
at least, if the three-form $\delta[G \mathrm{d}G^{-1}] \wedge G\mathrm{d}G^{-1} \wedge G\mathrm{d}G^{-1}$ is continuous inside $\Sigma$, according to the Leibniz integration rule
\cite[p.~466]{Courant_2}. Thus, to ensure $\delta \mathcal{N}_{3}=0$, it is necessary to have
\begin{equation}\label{variation_eq}
3\;N\; \int_{\partial\Sigma} \mbox{Tr}\left[K \; (G\delta G^{-1})  G\mathrm{d}G^{-1} \wedge G\mathrm{d}G^{-1}\right] \;=\; 0\;,
\end{equation}
where, as $\Sigma\equiv S^{3}$ in momentum space, $\partial\Sigma$ is the spherical two-dimensional surface with radius $|\vec{p}|$, whilst
$|\vec{p}|\to+\infty$. Eq.~\eqref{variation_eq} is the most general condition that any Green function $G$ must fulfill so that
$\mathcal{N}_{3}$ is invariant under small (continuous) deformations of $G$. In  Section \ref{Section_fermionic}, we shall derive a more specific condition for a particular, but general enough, class of fermionic
propagators. These have been used to fit lattice data for the quark propagator in QCD; see \cite{Furui:2006ks,Dudal:2013vha,Burgio:2012ph,Rojas:2013tza}.

A priori, one might question the physical relevance of the  topological invariants discussed here as they are explicitly based on gauge variant input, viz.~propagators that will depend on a chosen gauge. In practice, this
will boil down to a possible dependence on a gauge parameter. Indeed, changing the gauge parameter can always be achieved via consecutive application of small (infinitesimal) gauge variations (or BRST variation, if you
wish), which is just a special class of continuous deformations leaving the topological number untouched.  We will have nothing to say about large gauge transformations and invariance of ${\cal N}_3$ w.r.t.~those.

\section{Topological invariant for fermionic case: old and new results}
\label{Section_fermionic}

In this Section we will apply the previous ideas to the case of a fermion propagator in a $SU(N)$ gauge theory, with most general parameterization (assuming Lorentz invariance of course):
\begin{equation}\label{slrigh}
{G}^{AB}(p)=\mathcal{G}(p)\;\delta ^{AB}~=~\frac{Z(p^{2})}{i\slashed{p}+%
\mathcal{M}(p^{2})}\delta ^{AB}\;,
\end{equation}
where $AB$ are internal Dirac fermion indices\footnote{Usually, fermions belong to the fundamental representation of the internal $SU(N)$
group.}. It can be shown that the boundary condition \eqref{variation_eq} is recovered if \cite{Zubkov:2012ws}
\begin{equation}\label{condition_limit}
\lim_{|\vec{p}|\to+\infty} \; \frac{\delta\mathcal{M}}{|\vec{p}|\left(1+\frac{{\mathcal{M}}^{2}}{p^{2}}\right)^{2}}\;=\;0 \;.
\end{equation}
For this kind of fermionic propagators and up to powers of logarithms as dictated by the renormalization group equation, the dynamical mass $\mathcal{M}(p^{2})$ tends to a constant mass $\mu$ (possibly zero) in the UV
limit
($p^{2}\rightarrow \infty $), while the renormalization function $Z(p^{2})$ goes to $1$ in such a limit. Therefore, when
$p^{2}\rightarrow \infty$ the standard UV fermionic propagator $\frac{\delta^{AB}}{\ii\slashed{p}+\mu}$
is restored, up to logarithmic terms. On the other hand, in the IR limit ($p^{2}\rightarrow 0$) one usually assumes that $\mathcal{M}(p^{2})$ continuously tends to a
constant value $\mathcal{M}_{0}$. This can be appreciated from e.g.~lattice, functional or fitting approaches,
\cite{Braun:2014ata,Aguilar:2010cn,Alkofer:2008tt,Bashir:2012fs,Oliveira:2018lln,Bowman:2005vx,Burgio:2012ph,Dudal:2013vha,Furui:2006ks}.

Assuming the above mentioned asymptotic behavior of $Z(p^{2})$ and $\mathcal{M}(p^{2})$, one can verify directly that propagators of the kind
\eqref{slrigh} fulfil condition \eqref{condition_limit}.

Then, following \cite{Volovik:2009au}, the matrix element $K$ will be considered as being the one that accounts for the CT-symmetry, representable by
$\gamma^{5}\gamma^{0}$ (the four-dimensional Euclidean Dirac matrices). A nice overview of CPT symmetry can be found in \cite{CPT}. Notice that $\gamma_{5}\gamma_{0}$ does commute with the general Dirac fermion
propagator given at Eq.~\eqref{slrigh}. Then, since we already showed that ${\cal N}_3$ is a topological invariant, we can set $Z(p^2)=1$, as this $Z(p^2)$ for $p^2>0$ is a smooth deformation of $1$. One can again
appreciate this from non-perturbative functional or lattice computations; see e.g.~\cite{Braun:2014ata,Aguilar:2010cn,Alkofer:2008tt,Bashir:2012fs,Oliveira:2018lln,Bowman:2005vx,Burgio:2012ph,Furui:2006ks}. Therefore,
after some algebraic manipulations, we can reduce
\begin{eqnarray}  \label{fermion_top_inv}
{\mathcal{N}}_{3} &=& \frac{1}{24\pi^{2}} \varepsilon_{ijk} \mathrm{Tr}
\int d^{3}p \; \Bigg[ \gamma_{0}\gamma_{5} {\mathcal{G}} \left(
\partial_{p_{i}} {\mathcal{G}}^{-1} \right) {\mathcal{G}} \left(
\partial_{p_{j}} {\mathcal{G}}^{-1} \right) {\mathcal{G}} \left(
\partial_{p_{k}} {\mathcal{G}}^{-1} \right) \Bigg] \;,
\end{eqnarray}
to
\begin{eqnarray}
{\mathcal{N}}_{3} &=& \frac{4}{\pi} \int_{0}^{\infty} dp \; \frac{
p^{2} \left[ \mathcal{M}(p^{2}) - 2p^{2} \, \partial_{p^{2}}\mathcal{M}%
(p^{2}) \right] }{ \left[ p^{2} + \mathcal{M}^{2}(p^{2}) %
\right]^{2}} \;  \label{dhafdhf}
\end{eqnarray}
by making use of the Euclidean identities
\begin{eqnarray*}
\varepsilon_{ijk}\mathrm{Tr} \left[ \gamma^{5}\gamma^{0}\gamma_{i}\gamma_{j}\gamma_{k} \right]&=&-24 \;, \\
3p^{l}p^{k}\varepsilon_{ijk}\mathrm{Tr} \left[ \gamma^{5}\gamma^{0}%
\gamma_{l}\gamma^{i}\gamma^{j} \right] &=& -24\,\vec{p}\cdot\vec{p} \;.
\end{eqnarray*}
Let us now observe that Eq.~\eqref{dhafdhf} still carries the ``fingerprint'' of the topological invariance of the original expression \eqref{fermion_top_inv}. Indeed, even after performing the integral over the momentum
in order to arrive at Eq.~\eqref{dhafdhf}, this expression can be interpreted as a one-dimensional topological invariant, associated to the analytic structure of $\mathcal{M}(p^{2})$ itself. A fundamental observation is that one can also forget about how we did arrive at Eq.~\eqref{dhafdhf}, since such an expression is a perfectly well-defined winding number associated to the dynamical mass $\mathcal{M}(p^2)$ through the curve $\theta(p)$ defined in Eq.~\eqref{hgihgef} below, providing there is no pathological behavior of the integrand. In other words, if one is presented  the expression \eqref{dhafdhf} for the first time, without any previous knowledge of Eq.~\eqref{fermion_top_inv}, one is still able to show that this quantity \eqref{dhafdhf} is a well-defined one-dimensional topological invariant. More precisely,
\begin{eqnarray}
{\mathcal{N}}_{3} &=&
\frac{1}{2\pi}
\int_{0}^{2\pi} \, \mathrm{d}\theta
~=~ 1
\;,
\label{grrpef}
\end{eqnarray}
where we have defined the function $\theta(p)$ such that
\begin{eqnarray}
\theta(p) ~=~
4\left[
\frac{\mathcal{M}(p)/p}{1 + \left(\mathcal{M}(p)/p\right)^{2}} + \Arctan\left(\mathcal{M}(p)/p\right)
\right]
\;.
\label{hgihgef}
\end{eqnarray}
The significant advantage of this observation is that Eq.~\eqref{dhafdhf} can be considered as a legitimate invariant in itself for any dynamical (fermionic or bosonic) mass functions (that satisfy boundary conditions).

In the following we present a detailed analysis of the topological classification of a Dirac fermionic system whose dynamical mass reads
\begin{equation}
\mathcal{M}(p^{2})~=~\frac{M^{3}}{p^{2}+m^{2}}+\mu
\;.
\label{rbngoo}
\end{equation}
For the obtained values of parameters coming from a lattice fit \cite{Furui:2006ks,Dudal:2013vha,Rojas:2013tza}, $(M^{3},m^{2},\mu)$, (always keeping that order) all of them are positive and, therefore, the result for the topological invariant is always $\mathcal{N}_{3}=1$, as one can readily verify. However, depending on possible other configurations of these parameters, other values of ${\cal N}_3$ can emerge. This ``pathological'' behavior could imply phase transitions, provided in such a case $\delta\mathcal{N}_{3}=0$ cannot be guaranteed for a small variation of the parameters. This is translated into the fact that there are configurations, characterized by an equation on the space of parameters $f(M^{3},m^{2},\mu)=0$, which are unstable under small variation of the parameters. These surfaces can be realized as {\it phase boundaries} \cite{Goldenfeld}. One concrete example of such phase boundary is shown in Fig.~\ref{surface_fig}, where the two-dimensional surface is characterized by the equation $M^{3}+m^{2}\mu=0$ on the three-dimensional space of parameters\footnote{One can check that this condition translates into the Green function having a zero mass pole.}. On such surface the
topological invariant is $\mathcal{N}_{3}=0$. Notice that this is also the topological invariant's value in the chirally invariant case of a simple massless Dirac quark propagator.   As such, the topological number $\mathcal{N}_3$ can discriminate between chirally symmetric and chirally broken phases.

Away from this critical surface, the $\mathcal{N}_3$ jumps. To set the mind, working in unspecified mass units, for $(-1,1,1)$ we have $\mathcal{N}_3=0$;  for $(-1.01,1,1)$, $\mathcal{N}_3=-1$ and for $(-0.99,1,1)$, $\mathcal{N}_3=+1$. For all of these parameter sets, the propagator poles are located on the negative real axis, whilst the mass function remains positive for $p^2>0$. Stated otherwise, all of them could in principle be used to try to fit the non-perturbative quark mass function as found on lattice; see e.g.~\cite{Bowman:2005vx,Oliveira:2019erx}. In future work, it would be interesting to compute $\mathcal{N}_3$ directly based on lattice data, after which all kinds of approximations to it, for instance born out of Dyson-Schwinger/functional renormalization group equations mentioned already before, or approaches likes \cite{Alkofer:2003jj,Pelaez:2014mxa},  can be tested if they belong to the same homotopy class or not.

\begin{figure}[tbp]
\begin{center}
\includegraphics[width=1.0\textwidth,angle=0]{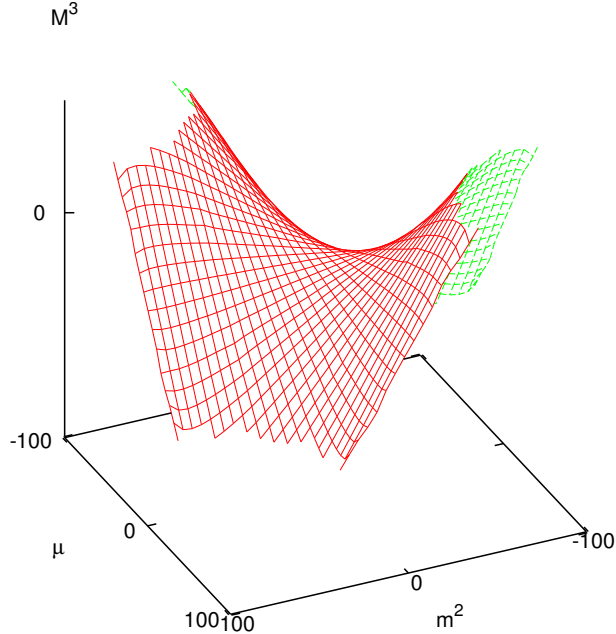}
\caption{A two-dimensional surface in the space of parameters where the topological invariant is $\mathcal{N}_{3}=0$, but any small perturbation will give a
nonzero result. In this case, the phase boundary is characterized by the equation $M^{3}+m^{2}\mu=0$.}
\label{surface_fig}
\end{center}
\end{figure}

More general, if the parameters take (critical) values such that the denominator of the integrand in Eq.~\eqref{dhafdhf} has singular behavior for positive values of $p^{2}$, there will be a transition $\Delta\mathcal{N}_{3}\neq0$ of the topological invariant for a small variation of the critical parameters. As
$\mathcal{N}_{3}$ is an integer number, this transition must come with another integer value for the topological number.


\section{A generalized topological invariant}
\label{Section_Generalized_TI}

In this Section, we propose to extend the momentum space topological invariant \eqref{dhafdhf} in order to make it sensitive to complex poles of the mass function $\mathcal{M}$, of which we always assume to know the analytic structure\footnote{Evidently, this is a big assumption. For example, a Monte Carlo based lattice computation will never give direct access to the full analytic structure in the complex momentum squared plane. But,
using dedicated inversion schemes can shed light on this anyhow, see e.g.~\cite{Fischer:2017kbq}, in some cases also (numerical) solutions over the complex momentum squared plane of the Dyson-Schwinger equations \cite{Maris:1995ns}\cite{Fischer:2017kbq} can be obtained. }. Such an extension is a topological invariant, applicable either to fermionic or to bosonic systems.

This quantity is defined by the following integral in the complex plane:
\begin{equation}
\mathcal{N}_{\Gamma }~=~\frac{1}{2\pi\ii }
\oint_{\Gamma }\mathrm{d}z\;
4\ii\sqrt{z}\frac{\left[
\mathcal{M}(z)-2z\,\partial _{z}\mathcal{M}(z)\right] }{\left( z+\mathcal{M}%
^{2}(z)\right) ^{2}}
\;,
\label{topological_invariant_def}
\end{equation}
where $\Gamma $ is the contour defined to enclose all the possible complex poles of the integrand \eqref{topological_invariant_def} lying in the upper half
of the complex plane (see Fig.~\ref{complex_plane_fig}), while avoiding every possible real pole, branch point and branch cut. This new topological object is a relative of the topological object $N_W$ defined in \cite{Hayashi:2018giz}, with the difference that their object is defined in terms of the propagator itself, while our ${\cal N}_{\Gamma}$ depends on the mass function ${\cal M}(z)$, which in general can be defined by writing in full generality a fermion propagator as in \eqref{slrigh}, or for a bosonic propagator (stripping off all possible color/Lorentz tensorial structures)
\begin{equation}\label{dd1}
  D(p^2)=\frac{Z(p^2)}{p^2+{\cal M}^2(p^2)}
\end{equation}
where the same comments as before apply to the wave-function normalization function $Z(p^2)$ and mass function ${\cal M}(p^2)$.

Despite the similarity, it is not clear to us if both objects encode the same physical information. Furthermore, notice that the integrand of Eq.~\eqref{topological_invariant_def} has at least two branch points, $0$ and $\infty$, but depending on the analytical structure of ${\cal M}(z)$, the number of non-removable singularities can increase. Likewise, it seems that also the quantity defined in \cite{Hayashi:2018giz} is susceptible to the
appearance of new branch points, regarding the analytical expression of the propagator.

Therefore, we choose the closed contour $\Gamma$ as being composed of three paths with the following parametrization (see Fig.~\ref{complex_plane_fig}):
\begin{eqnarray}
&&
C(\alpha) = R\e^{i \alpha} \qquad \text{with} \qquad \pi - \frac{\varepsilon}{R} \leq \alpha \leq \frac{\varepsilon}{R}
\label{paramtr1}
\\
&&
L_{-}(t) = t + i\varepsilon
\qquad \text{with} \qquad  -R \leq t < 0
\label{paramtr2}
\\
&&
L_{+}(t) = t + i\varepsilon
\qquad \text{with} \qquad   0 \leq t \leq R
\label{paramtr3}
\;.
\end{eqnarray}
It is important to notice that Eq.~\eqref{topological_invariant_def} can be rewritten as the winding
number of some function $f(z)$, indeed,
\begin{eqnarray}
{\cal N}_{\Gamma} ~=~
\frac{1}{2\pi\ii}
\oint_{\Gamma} dz\; \frac{f'(z)}{f(z)}
\label{topinvhsigh}
~=~
-\frac{1}{2\pi}
\int_{\theta\left( f(\Gamma) \right)}
d\theta(z)
\label{wirgwigh}
\end{eqnarray}
with
\begin{eqnarray}
f(z) ~=~
\e^{- \ii\theta(z)}
\quad \;, \quad
\theta(z) ~=~
2\left[
\frac{\sqrt{z}{\cal M}(z)}{z + {\cal M}^{2}(z)} + \Arctan\left( \frac{{\cal M}(z)}{\sqrt{z}}\right)
\right]
\label{lrgrihg}
\;.
\end{eqnarray}
Notice that the $\Arctan$-term gives the relevant contribution to the integral \eqref{wirgwigh}, in a sense that an equally well defined topological object \eqref{wirgwigh} could be given by $\theta(z) =2 \Arctan\left( \frac{{\cal M}(z)}{\sqrt{z}}\right)$. From Eq.~\eqref{wirgwigh} one can clearly see that our ${\cal N}_{\Gamma}$ is the winding number of $f(z)$ given at Eq.~\eqref{lrgrihg}, and as such it represents the difference between the number of zeros and poles of $f(z)$ as the latter function is per assumption meromorphic within the contour $\Gamma$ and has no poles or zeros on the contour $\Gamma$.

\subsection{Invariance of $\mathcal{N}_{\Gamma}$ under small deformations of the parameters}

Now, in order to derive the boundary condition for ${\cal N}_{\Gamma}$ to be invariant under smooth variations of ${\cal M}$ (with respect to the
parameters of ${\cal M}$), we will rely on two procedures. In this Subsection, we use general complex analysis results, while the second argument, $cf.$ Appendix \ref{app}, makes direct use of expressing Eq. \eqref{topological_invariant_def} through the function $\theta(z)$ defined in Eq. \eqref{lrgrihg}.

If we consider the variation of $\mathcal{N}_{\Gamma}$, we get\footnote{The contour $\Gamma$ itself depends on the integrand, in the sense that $\Gamma$ is chosen in such way that it does not encircle or pierce any cuts or
has poles on it. As before, this means that the renormalization group related logarithms affecting the integrand are analytic within $\Gamma$ and as such can be considered as smooth deformations. Stated otherwise, these
logarithmic terms can again be neglected being irrelevant for determining the topological number ${\cal N}_\Gamma$.}
\begin{eqnarray*}
  \delta\mathcal{N}_{\Gamma} &=& \frac{2}{\pi} \oint_{\Gamma }\mathrm{d}z\;   \sqrt{z}\frac{\left[
\delta\mathcal{M}-2z\,\partial _{z}\delta\mathcal{M}\right] }{\left( z+\mathcal{M}^{2}\right) ^{2}}
- \frac{8}{\pi} \oint_{\Gamma }\mathrm{d}z\; \sqrt{z}\frac{\left[ \mathcal{M}-2z\,\partial _{z}\mathcal{M}\right]\mathcal{M}\delta\mathcal{M} }{\left(z+\mathcal{M}%
^{2}\right) ^{3}}
\;,
\end{eqnarray*}
where $\mathcal{M}$ is a function of $z$; we did not write it explicitly to avoid notational clutter. In order to have something proportional to $\delta\mathcal{M}$, we need to integrate by parts the term with
$\partial_{z}\delta\mathcal{M}$. This is achieved via
\begin{equation*}
-\frac{2\sqrt{z}z}{(z+\mathcal{M}^{2})^{2}} \; \partial_{z}\delta\mathcal{M} = \frac{3\sqrt{z}\delta\mathcal{M}}{(z+\mathcal{M}^{2})^{2}} -
\frac{4\sqrt{z}z(1+2\mathcal{M}\partial_{z}\mathcal{M})\delta\mathcal{M}}{(z+\mathcal{M}^{2})^{3}} - \frac{\partial}{\partial z}\left(\frac{2\sqrt{z}z\delta\mathcal{M}}{(z+\mathcal{M}^{2})^{2}}\right) \;.
\end{equation*}
Plugging this directly  into $\delta\mathcal{N}_{\Gamma}$, we obtain
\begin{equation}\label{variationNGamma}
  \delta\mathcal{N}_{\Gamma} \;=\; -\frac{4}{\pi} \oint_{\Gamma }\; \frac{\partial}{\partial z}\left(\frac{\sqrt{z}z\delta\mathcal{M}}{(z+\mathcal{M}^{2})^{2}}\right)\;\mathrm{d}z\;.
\end{equation}
The issue to prove that $\delta\mathcal{N}_{\Gamma}=0$ is the following: we must ascertain that the function $\frac{\partial}{\partial z}\left(\frac{\sqrt{z}z\delta\mathcal{M}}{(z+\mathcal{M}^{2})^{2}}\right)$ has a zero residue inside $\Gamma$, otherwise the zero value of the integral cannot be guaranteed. Therefore we consider that $\mathcal{M}$ has the form of a rational function (fermionic case), or the square root of some rational
polynomial (bosonic case); we thus ignore, without any loss of generality, the logarithmic tails, as explained before. Then $\frac{\sqrt{z}z\delta\mathcal{M}}{(z+\mathcal{M}^{2})^{2}}$ is a meromorphic function times a possible square root of some rational polynomial. As the contour $\Gamma$ explicitly avoids the branch points and branch cuts coming from such a square root, we can say that $\frac{\sqrt{z}z\delta\mathcal{M}}{(z+\mathcal{M}^{2})^{2}}$ is a meromorphic function inside the contour $\Gamma$. Therefore, we are under the hypothesis of Lemma \ref{lemma1}, at least inside $\Gamma$. A direct result of the lemma, with the proof sketched below, is that our $\frac{\partial}{\partial z}\left(\frac{\sqrt{z}z\delta\mathcal{M}}{(z+\mathcal{M}^{2})^{2}}\right)$ indeed has a zero residue inside $\Gamma$.
\begin{lemma}\label{lemma1}
Let $F(z)$ be a complex function with a pole at $z=a$ with finite multiplicity $m$ and analytic in the punctured disc $D'(a,r)=\{z\in\mathbb{C},\;0<|z-a|<r\}$. \\
Then $F'(z)$ has a pole with multiplicity $(m+1)$ and zero residue at $z=a$.
\end{lemma}
\begin{proof}
As the multiplicity of the pole is $m$, we may consider the Laurent series of $F(z)$ in $D'(a,r)$ (see for instance \cite{Ablowitz_Fokas})
\begin{equation*}
F(z)=\sum\limits_{N=-m}^{+\infty} C_{n}(z-a)^{n}=\sum\limits_{n=1}^{m} \frac{C_{-n}}{(z-a)^{n}}+\sum\limits_{n=0}^{+\infty} C_{n}(z-a)^{n}\;.
\end{equation*}
The coefficient $C_{-1}$ is the residue of $F(z)$ at $z=a$. Therefore,
\begin{eqnarray*}
F'(z)&=&-\sum\limits_{n=1}^{m} \frac{nC_{-n}}{(z-a)^{n+1}}+\sum\limits_{n=1}^{+\infty} nC_{n}(z-a)^{n-1}
    =\sum\limits_{n=2}^{m+1} \frac{B_{n}}{(z-a)^{n}}+\sum\limits_{n=0}^{+\infty} D_{n}(z-a)^{n} \;,
\end{eqnarray*}
where $B_{n}= -(n-1)C_{-n+1}$ and $D_{n}= (n+1)C_{n+1}$. Thus
\begin{equation}\label{Laurent_f'}
F'(z)=\sum\limits_{n=-m-1}^{+\infty} E_{n}(z-a)^{n}\;,
\end{equation}
with $E_{n}=B_{-n}$ for $n\leq-2$, $E_{n}=D_{n}$ for $n\geq0$ and $E_{-1}=0$. Eq. \eqref{Laurent_f'} is, by construction, the Laurent series of $F'(z)$ at $z=a$. Then, $z=a$ is a pole of  multiplicity $(m+1)$, as the series
starts with a term proportional to $(z-a)^{-m-1}$. It also has a residue equal to zero, as $E_{-1}=0$.
\end{proof}

\section{Applications}
\label{Section_Applications}

In this Section we will apply our complex topological object \eqref{topological_invariant_def} on both fermionic and gluonic sectors as a matter of exercise, but, being the form of the propagator for Dirac fermions at finite temperature more complicated because of the absence of Lorentz invariance, the Dirac fermions will be treated at zero temperature while the gluonic case, only, will be treated at finite temperature.

\subsection{A Dirac quark propagator that fits lattice data}

Here we assume the analytic continuation $p^{2} \to z$ of the Dirac quark propagator that fits lattice data, whose $\mathcal{M}(p^{2})$ mass function is given by Eq. \eqref{rbngoo}. The boundary conditions \eqref{bound.cond.} for $\delta {\cal N}_{\Gamma} = 0$ are satisfied, namely,
\begin{eqnarray}
\lim_{|z|\to\infty} {\cal M}(z) ~=~ \mu
\qquad \text{and} \qquad
\lim_{|z|\to\infty} \delta {\cal M}(z) ~=~ \delta \mu
\;.
\end{eqnarray}
We must emphasize that, because of the shape of our contour $\Gamma$, every possible complex pole with positive real part lies within the surrounded region (see Fig. \ref{complex_plane_fig}). For the propagator on the
right hand side of Eq. \eqref{slrigh} and
with dynamical mass given by Eq. \eqref{rbngoo}, there will always be three poles, one of these being negatively real, and two complex conjugate ones. Notice that
for the present quark propagator there will not appear new branch points due to the specific structure of ${\cal M}(z)$.
\begin{figure}[tbp]
\begin{center}
\includegraphics[width=0.7\textwidth,angle=0]{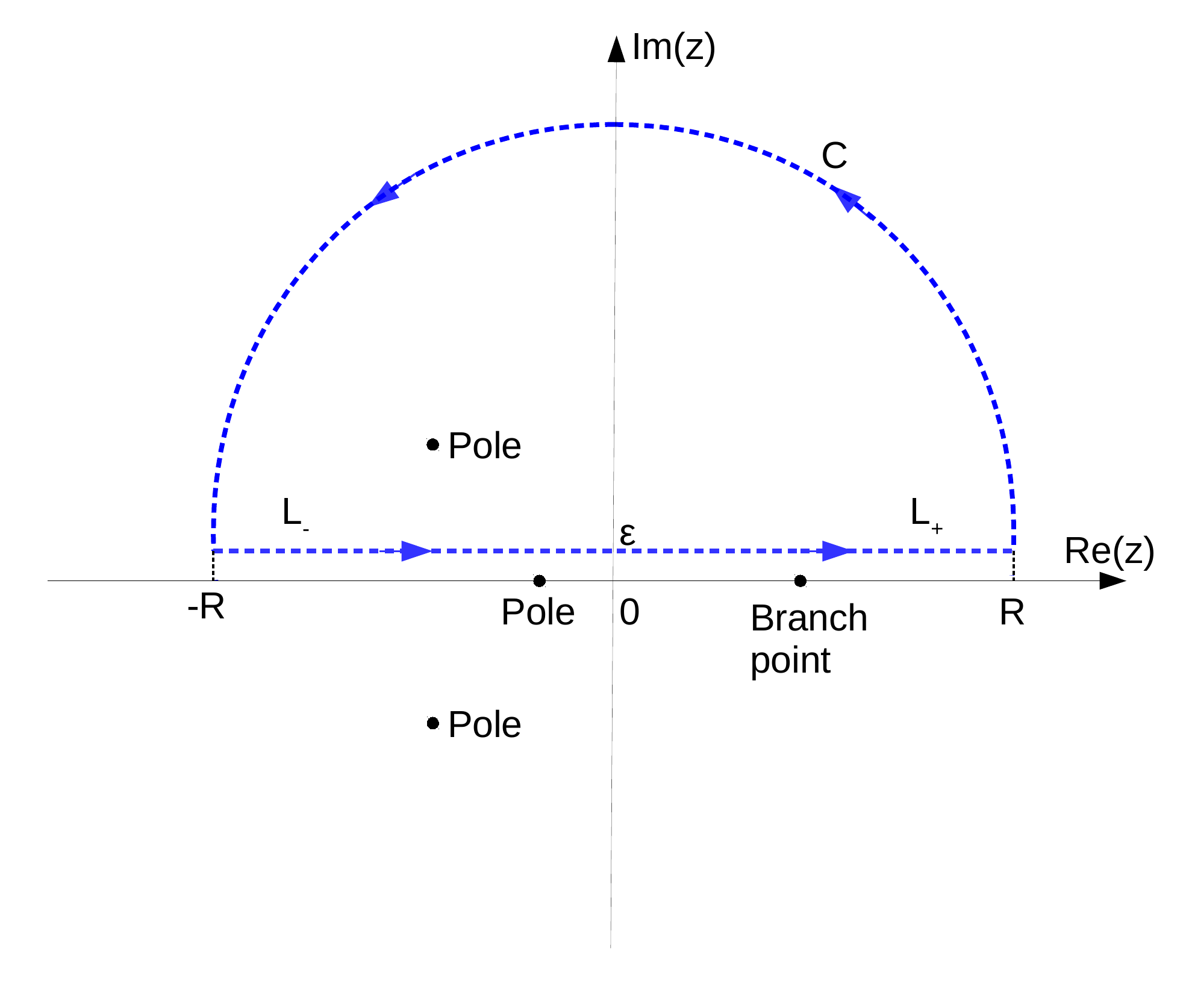}
\end{center}
\caption{A sketch of the procedure to obtain $\mathcal{N}_{\Gamma }$ in the complex plane. Here we consider a propagator displaying three poles (one real and two complex conjugates), and a possible new branch point that may appear. We choose to close the contour in the upper half of the complex plane in a way that every possible real pole and the branch cut is avoided. Thus, the contour $\Gamma$ is defined by the semi-circle $C$ of radius $R$ plus the horizontal
paths $L_{-}(z)$ and $L_{+}(z)$, which are lifted up by an imaginary infinitesimal, $i\varepsilon$.}
\label{complex_plane_fig}
\end{figure}

For general values of the parameters $(M^{3},m^{2},\mu)$, we find ${\cal N}_{\Gamma} = -2$ and ${\cal N}_{\Gamma} = +2$. More precisely, ${\cal N}_{\Gamma}=-2$  was associated to a negative mass at zero momentum, and also with a massless quark at zero momentum evolving to a strictly negative mass. It is possible to find combinations of the parameters so that the mass function is divergent at some point. In these cases, if the mass starts from negative values (considering $p^{2}$ from zero to positive values), ${\cal N}_{\Gamma}$ also acquires the value $-2$. On the other hand, we verified that ${\cal N}_{\Gamma}=+2$ is associated to configurations of the parameters such that the dynamical mass is positive at zero momentum, and also if it is a massless quark at zero momentum with an increasing positive mass. For example, for the specific values of $(M^{3},m^{2},\mu)$ that makes the quark propagator fit to lattice data, \cite{Furui:2006ks,Dudal:2013vha}, our topological invariant acquires the positive value $+2$.

\subsection{A gluon propagator that fits lattice data}

In order to explicitly show that our generalized topological object \eqref{topological_invariant_def} can in fact be applied to the topological analysis
of bosonic relativistic quantum fields within the MST framework, we perform here the topological classification of the space defined by the two-point Green function of gluons, whose analytic expression is one that has been used in literature to fit lattice data at zero and finite temperature, according to \cite{Aouane:2011fv,Cucchieri:2012gb,Dudal:2013vha,Dudal:2018cli}. Such analytic expression can be derived by means of the refined Gribov-Zwanziger framework, see e.g. \cite{Dudal:2008sp,Dudal:2019ing}.

Assuming the particular notation of \cite{Aouane:2011fv}, the gluon propagator reads,
\begin{eqnarray}
D(p^{2}) ~=~
\frac{ c(1 + d\,p^{2})}{p^{4} + 2r^{2}p^{2} + r^{4} + b^{2}} ~=~
 \frac{cd}{p^{2} + {\cal M}^{2}(p^{2})}
\;.
\label{gluonprop}
\end{eqnarray}
Performing the analytic continuation $p^{2}\to z$ to the complex plane of the mass function, ${\cal M}$, one has
\begin{eqnarray}
{\cal M}^{2}(z) ~=~
\frac{ z(2r^{2}d - 1) + (r^{4} + b^{2})d}{1 + dz}
\;.
\label{easdgh4g}
\end{eqnarray}
Notice that the boundary conditions \eqref{bound.cond.} are satisfied,
\begin{eqnarray}
\lim_{|z| \to \infty}{\cal M}(z) ~=~ \sqrt{\frac{ 2r^{2}d - 1}{d}}
\quad \text{and} \quad
\lim_{|z| \to \infty}\delta{\cal M}(z) ~=~ \frac{4rd^{\frac{5}{2}} \, \delta r + \mathrm{d}^{\frac12} \, \delta d}{d^{2}(2r^{2}d -1)^{\frac12}}
\;,
\end{eqnarray}
so that ${\cal N}_{\Gamma}$ is indeed a topological invariant under smooth variation of ${\cal M}$ with respect to the parameters $(r,\ b,\ d,\ c)$.

As mentioned before, since our proposed ${\cal N}_{\Gamma}$ is constructed in the complex plane and explicitly depends on the expression of the
(analytic continuation of the) mass function of the propagator ${\cal M}(p^{2})$, it is quite possible that non-removable singularities appear, others
than $0$ and $\infty$, as in the case of the present example. Given the expression of the complex mass function ${\cal M}(z)$, at Eq. \eqref{easdgh4g}, one can check that the integrand of
Eq. \eqref{topological_invariant_def} has also the new branch points $-\frac{1}{d}$ and
$\frac{-b^2 d - d r2^2}{-1 + 2 d r^2}$, which are real and, therefore, lie outside the region bounded by $\Gamma$.

In this example we restrict ourselves to the finite temperature results of \cite{Aouane:2011fv}. The values of $(r,~b,~d,~c)$ are taken from
their Table II. Namely, the authors of \cite{Aouane:2011fv} assessed the following values of temperature (in unities of the critical temperature):
$T/T_{c} = 0.65$; $T/T_{c} = 0.74$; $T/T_{c} = 0.86$; $T/T_{c} = 0.99$; $T/T_{c} = 1.20$; $T/T_{c} = 1.48$; $T/T_{c} = 1.98$; and $T/T_{c} = 2.97$. For
all the temperature values with the exception of the highest three, they set $b=0.0$. Notice that for $T/T_{c} = 1.20$ (clearly above the critical
temperature) the parameter $b$ is also set to zero.

For the values reported in \cite{Aouane:2011fv}, our topological invariant acquires the following values: ${\cal N}_{\Gamma} = 0$ whenever
$b=0$ , since the propagator then displays only a (negative) real (double) pole. Such configuration appears for temperatures $T/T_{c} = 0.65$;
$T/T_{c} = 0.74$; $T/T_{c} = 0.86$; $T/T_{c} = 0.99$; $T/T_{c} = 1.20$. On the other hand, we find ${\cal N}_{\Gamma} = +2$ for the three highest temperatures $T/T_{c} = 1.48$;
$T/T_{c} = 1.98$; and $T/T_{c} = 2.97$ (where $b \neq 0$). Therefore, these results suggest a  phase transition at a critical temperature
somewhere between $T/T_{c} = 1.20$ and $T/T_{c} = 1.48$, different from the deconfinement phase transition found by the authors of
\cite{Aouane:2011fv}.

Interestingly, the authors of \cite{Hayashi:2018giz} found $N_W = 0$ for gauge propagators displaying only real poles; and $N_W = -2$ for gauge propagators with a set of complex conjugate poles. Here, again, we point out
that our topological invariant assumes the same absolute values as the one of
\cite{Hayashi:2018giz}, despite that  they do not necessarily represent  the same topological quantity. Therefore, following a completely different approach from the one in \cite{Hayashi:2018giz} (although both are aimed
to the momentum space topological analysis of the gauge field two-point Green's function), we could find a  phase transition between different regimes of the gauge
propagator, in agreement with \cite{Hayashi:2018giz}.

\section{Conclusions}
\label{Section_Conclusions}

We developed a topological classification of relativistic quantum systems within the momentum space topology framework, starting with the usual
topological invariant ${\cal N}_{3}$ as applied to Dirac fermions. This topological invariant can experience integer jumps when one crosses different open regions, which are limited by \emph{phase boundaries}. This change can be related to the non-analyticities appearing in the integrand of \eqref{dhafdhf}.

Afterwards, we proposed a new topological invariant, ${\cal N}_{\Gamma}$, that is sensitive to the existence of complex poles of the propagator of the
relativistic quantum field, be it fermion or boson. This topological invariant was constructed by performing the analytic continuation $p^{2} \to z$ to the complex plane of the usual ${\cal N}_{3}$. As a result, ${\cal N}_{\Gamma}$ depends on the closed contour $\Gamma$ and on the analytic expression of ${\cal M}(z)$. In order to ensure the invariance of ${\cal N}_{\Gamma}$ with respect to smooth variations of ${\cal M}(z)$, we found that ${\cal M}(z)$ and $\delta {\cal M}(z)$ must be constant in the limit $|z| \to \infty$. We applied this topological number ${\cal N}_{\Gamma}$ to a model Dirac quark propagator whose
analytic expression can fit quite well corresponding lattice data. In addition, we also applied it to a gluon propagator form that can fit zero and finite temperature lattice data. We made use of the rational function fits according to \cite{Aouane:2011fv}, and we could find a phase transition at a temperature within the range $T/T_{c} = 1.20$ and $T/T_{c} = 1.48$, which is slightly above the thermodynamic critical (deconfinement) temperature found by the authors.

While Eq. \eqref{topinv} has the clear physical meaning of being related to a non-dissipative conductivity (Hall, spin Hall, etc,  $cf.$ \cite{Zubkov:2019amq,Zhang:2019zqa}) at this stage the physical interpretation of our invariant is less obvious. What is clear is that it is sensitive to the presence of complex conjugate poles thus, as shown in refs.~\cite{Hayashi:2018giz} and \cite{Cyrol:2018xeq}, it ``measures'' the possible occurrence of a violation of the ``reflection positivity condition''  (the Euclidean counterpart to the positive definiteness of the norm in the Hilbert space of the corresponding Wightman quantum field theory, see the Osterwalder-Schrader axioms \cite{Osterwalder:1973dx,Osterwalder:1974tc}) in terms of the Schwinger function. Violation of the reflection positivity condition is regarded as a signal for gluon confinement (see Section VI of \cite{Hayashi:2018giz} and Section V and VI of \cite{Cyrol:2018xeq}, as well as \cite{Krein:1990sf}.).

For the particular propagator \eqref{gluonprop}, depending on the value of the parameters, we will encounter either two (possibly degenerate) real poles or two complex conjugate ones. Only in the latter case, we will have a contribution to the topological invariant as defined in expression \eqref{topological_invariant_def}. Moreover, in the presence of a set of such poles, it has been shown explicitly in e.g.~\cite{Dudal:2008sp,Roberts:1994dr,Gao:2017uox} that the Schwinger function,
\begin{equation}\label{schwinger}
  C(t)=\frac{1}{2\pi} \int_{-\infty}^{+\infty} dp e^{ipt} D(p^2)\,,
\end{equation}
is no longer positive definite, in accordance with lattice measurement of the same quantity, see for example \cite{Cucchieri:2004mf,Bowman:2007du}.

For these reasons we believe that our topological invariant proposed in equation \eqref{topological_invariant_def} may be interpreted as an indication of whether the quantum system is in a confined or deconfined phase. However, for the moment, this is only an \textit{intuition} without a rigorous analytic proof, which deserves further investigation.

For future work, it would be interesting to use improved fits to more recent finite temperature gluon data, to see if the phase transition coincides with the deconfinement transition. That such is possible can be inferred from the preliminary fitting report of \cite{Cucchieri:2012gb}, where in contrast to the here used numbers of \cite{Aouane:2011fv}, there are  complex poles below (up to $T=0$) and around $T_c$. More takes on the analytic structure of gluon and/or quark propagators can be found in $e.g.$ \cite{Dudal:2013yva,Siringo:2016jrc,Lowdon:2017gpp,Silva:2017feh,Ilgenfritz:2017kkp,Cyrol:2018xeq,Kondo:2019rpa,Dudal:2019gvn,Binosi:2019ecz}. Although these rational function fits never capture the renormalization group-controlled logarithmic tails, we explained how our construct is insensitive to these structures anyhow.

In addition, it would be rather interesting to follow the temperature-evolution, if any, of the quark topological number based on the rather recent lattice output of \cite{Oliveira:2019erx}, to see if one can find a topological signal of the deconfinement and/or chiral transition in the quark sector. We hope to come back to these issues in the foreseeable future.

\section*{Acknowledgments}

P.P.~, D.D.~and L.R.~acknowledge the warm hospitality of the Centro de Estudios Cient\'ificos (CECs) in Valdivia, Chile, during different stages of this work. P.~P. is supported by the project High Field Initiative (CZ.02.1.01/0.0/0.0/15\_003/0000449) from the European Regional Development Fund. I.~F.~J. is supported by Fondecyt Grant No. 3170278. This work was partially supported by the research Grant Number 2017W4HA7S ``NAT-NET: Neutrino and Astroparticle Theory Network'' under the program PRIN 2017 funded by the Italian Ministero dell' Istruzione, dell' Universit\`a e della Ricerca (MIUR) and also by the Initiative ``TAsP: Theoretical Astroparticle Physics'' funded by the Italian Istituto Nazionale di Fisica Nucleare (INFN). This work has been funded by the Fondecyt Grant No. 1160137. The Centro de Estudios Cient\'ificos (CECs) is funded by the Chilean Government through the Centers of Excellence Base Financing Program of Conicyt.

\appendix

\section{Alternative proof that $N_\Gamma$ is a homotopic invariant}\label{app}

For a second, more explicit, demonstration of the invariance of Eq.~\eqref{topological_invariant_def}, we split our contour as $\Gamma =
C(\alpha) + L_{-}(t) + L_{+}(t)$ at the limit $R \to \infty$, as depicted in Fig.~\ref{complex_plane_fig}, writing ${\cal
N}_{\Gamma}$ as
\begin{eqnarray}
{\cal N}_{\Gamma} ~=~
-\frac{1}{2\pi}
\lim_{R \to \infty}
\left[
\int_{ C (\alpha) }
\mathrm{d}\theta(z)
+
\int_{ L_{-}(t) }
\mathrm{d}\theta(z)
+
\int_{ L_{+}(t) }
\mathrm{d}\theta(z)
\right]
\label{gh2ngso}
\,.
\end{eqnarray}
Given the parameterizations \eqref{paramtr1}, \eqref{paramtr2} and \eqref{paramtr3}, one has,
\begin{equation}\label{gh2ngso2}
{\cal N}_{\Gamma} ~=~
-\frac{1}{2\pi}
\Bigg[
\int_{ \pi - \frac{\varepsilon}{R} }^{\frac{\varepsilon}{R}}
\frac{\mathrm{d}}{\mathrm{d}\alpha}\theta ( R\e^{i\alpha} )
\, \mathrm{d}\alpha
+
\int_{ -R }^{ 0 }
\frac{\mathrm{d}}{\mathrm{d}t}\theta(t + i\varepsilon)
\, \mathrm{d}t
+
\int^{ R }_{ 0 }
\frac{\mathrm{d}}{\mathrm{d}t}\theta(t + i\varepsilon)
\, \mathrm{d}t
\Bigg]
\;,
\end{equation}
with $R \to \infty$. Let us focus on the integral over the contour piece $C(\alpha)$. As we avoid poles and branch points/cuts on $\Gamma$, such an integral reduces to
\begin{eqnarray*}
-\frac{1}{2\pi} \, \lim_{\R\to\infty}  \theta ( R\e^{i\alpha} )  \Bigg\vert_{\pi - \frac{\varepsilon}{R}}^{\frac{\varepsilon}{R}}
&=&
-\frac{1}{\pi}
\lim_{R\to\infty}
\left[
\frac{R^{\frac12}\e^{\ii\frac{\alpha}{2}}{\cal M}(R\e^{\ii\alpha})}{R\e^{\ii\alpha} + {\cal M}^{2}(R\e^{\ii\alpha})}
+ \Arctan\left( \frac{{\cal M}(R\e^{\ii\alpha})}{R^{\frac12}\e^{\ii\frac{\alpha}{2}}}\right)
\right]\Bigg\vert^{\frac{\varepsilon}{R}}_{\pi - \frac{\varepsilon}{R}} \\
&=& 0
\;,
\end{eqnarray*}
if ${\cal M}(z)$ is at most constant at $|z| \to \infty$. This means that our topological invariant reduces to
\begin{eqnarray*}
{\cal N}_{\Gamma} ~=~
-\frac{1}{2\pi}
\lim_{R \to \infty}
\left[
\int_{ L_{-}(t) }
\mathrm{d}\theta(z)
+
\int_{ L_{+}(t) }
\mathrm{d}\theta(z)
\right]
\,.
\end{eqnarray*}
The variation $\delta_{_{\cal M}} {\cal N}_{\Gamma}$, due to $\delta {\cal M}(z)$, on Eq.~\eqref{gh2ngso2}, lead us to
{\small
\begin{eqnarray*}
\delta_{_{\cal M}} {\cal N}_{\Gamma}  ~=~
-\frac{1}{2\pi}
\Bigg[
\int_{\frac{\varepsilon}{R}}^{\pi - \frac{\varepsilon}{R}}
\frac{\mathrm{d}}{\mathrm{d}\alpha} \delta_{_{\cal M}}  \theta(R\e^{i\alpha}) \, \mathrm{d}\alpha
&+&
\int_{ -R }^{ 0 }
\frac{\mathrm{d}}{\mathrm{d}t} \delta_{_{\cal M}}  \theta(t + i\varepsilon)
\, \mathrm{d}t
-
\int^{ R }_{ 0 }
\frac{\mathrm{d}}{\mathrm{d}t} \delta_{_{\cal M}}  \theta(-t + i\varepsilon)
\, dt
\Bigg]
\,.
\end{eqnarray*}
}
The variation $\delta_{_{\cal M}} \theta(z)$ reads
\begin{eqnarray}
\delta_{_{\cal M}}  \theta(z) ~=~
4 \left[
\frac{\delta {\cal M}}{\sqrt{z}  \left( 1+ \frac{{\cal M}^{2}}{z}  \right)}
\right]
\;,
\end{eqnarray}
just as Eq.~\eqref{variationNGamma}. Thus, as there are no singularities along the contour $\Gamma$, in order to ensure $\delta_{_{\cal M}} {\cal N}_{\Gamma} = 0$, the analytical continuation of the mass function ${\cal
M}(z)$ and its small variation $\delta {\cal M}(z)$ must be constant at the boundary $|z| \to \infty$, i.e.,
\begin{eqnarray}
\lim_{|z|\to\infty} {\cal M}(z) ~=~ M_{\infty}
\qquad \text{and} \qquad
\lim_{|z|\to\infty} \delta{\cal M}(z) ~=~ \delta M_{\infty}
\;.
\label{bound.cond.}
\end{eqnarray}

\bibliographystyle{apsrev4-1}
\bibliography{top_inv_biblio}

\end{document}